\documentclass[a4paper,american]{llncs}
\usepackage[T1]{fontenc}
\usepackage{amsmath}
\usepackage{graphicx}

\makeatletter

\providecommand{\tabularnewline}{\\}

\usepackage[draft]{hyperref}
\usepackage{times}

\usepackage{babel}
\makeatother
\begin{document}

\title{Construction of Large Constant Dimension Codes With a Prescribed
Minimum Distance.}

\author{Axel Kohnert and Sascha Kurz}

\institute{Mathematisches Institut\\
University of Bayreuth\\
95440 Bayreuth\\
Germany\\
axel.kohnert@uni-bayreuth.de\\
sascha.kurz@uni-bayreuth.de}

\maketitle
\today

\begin{abstract}
In this paper we construct constant dimension space codes with prescribed
minimum distance. There is an increased interest in space codes since
a paper \cite{koetter-2007} by Kötter and Kschischang were they gave
an application in network coding. There is also a connection to the
theory of designs over finite fields. We will modify a method of Braun,
Kerber and Laue \cite{Braun_kerber_laue_q_analoga} which they used
for the construction of designs over finite fields to do the construction
of space codes. Using this approach we found many new constant dimension
spaces codes with a larger number of codewords than previously known
codes. We will finally give a table of the best found constant dimension
space codes.
\end{abstract}
network coding, q-analogue of Steiner systems, subspace codes

\section{Introduction}

\subsection{Subspace Codes}

In \cite{koetter-2007} R. Kötter and F. R. Kschischang developed
the theory of subspace codes for applications in network coding. We
will recapitulate their definitions in a slightly different manner.
We denote by $L(GF(q)^{v})$ the lattice of all subspaces of the space
of dimension $v$ over the finite field with $q$ elements. A \emph{subspace
code} $C$ is a subset of $L(GF(q)^{v})$. If all the subspaces in
$C$ are of the same dimension then $C$ is a \emph{constant dimension
code}.

The subspace distance between two spaces $V$ and $W$ in $L(GF(q)^{v})$
is defined as\[
d_{S}(V,W):=dim(V+W)-dim(V\cap W)\]
which is equal to\[
dim(V)+dim(W)-2dim(V\cap W).\]

This defines a metric on $L(GF(q)^{v})$. The minimum (subspace) distance
of a subspace code $C$ is defined as\[
D_{S}(C):=min\{ d_{S}(V,W):V,W\in C\mbox{\,\, and }V\not=W\}.\]

We define now the optimal (subspace) code problem: 

\begin{quotation}
(P1) For a given lattice $L(GF(q)^{v})$ fix a minimum (subspace)
distance $d$ and find the maximal number $m$ of subspaces $V_{1},\ldots,V_{m}$
in $L(GF(q)^{v})$ such that the corresponding subspace code $C=\{ V_{1},\ldots,V_{m}\}$
has at least minimum distance $d.$
\end{quotation}
The following point of view is useful for the study of subspace codes:
One of the classical problems in coding theory can be stated as follows: 

\begin{quotation}
(P2) Given the Hamming graph of all words of length $v$ and a minimum
distance $d$ find a maximal number $m$ of words such that the pairwise
minimum distance is at least $d.$ 
\end{quotation}
If we substitute the Hamming graph by the Hasse diagram of $L(GF(q)^{v})$
the problem (P2) becomes problem (P1). Both problems are special cases
of a packing problem in a graph. If we start with problem (P2) and
use the 'field with one element' we get the problem (P1). Because
of this property we say (P2) is the $q-$analogue of (P1). This connection
is well known (e.g. \cite{ahlswede_2001,etzion_qdesigns_2002}) and
will be useful in the following. Since the publication of the paper
by Kötter and Kschischang the constant dimension codes as the $q-$analogue
of the constant weight codes were studied in a series of papers \cite{gadouleau-2008,etzion-silberstein-2008,xia-2007}.

\subsection{$q-$Analogues of Designs}

A $t-(v,k,\lambda)$ design is a set $C$ of $k-$element subsets
(called blocks) of the set $\{1,\ldots,v\}$ such that each $t-$element
subset of $\{1,\ldots,v\}$ appears in exactly $\lambda$ blocks.
The special case of $\lambda=1$ is called a Steiner system.

The same construction which was used to connect problem (P1) to (P2)
in the subsection above can be used to define the $q-$analogue of
a $t-$design. A $t-(v,k,\lambda)$ design over the finite field $GF(q)$
is a multiset $C$ of $k-$dimensional subspaces (called $q-$blocks)
of the $v-$dimensional vector space $GF(q)^{v}$ such that each $t-$dimensional
subspace of $GF(q)^{v}$ is subspace of exactly $\lambda$ $q-$blocks. 

The connection with the constant dimension codes is given by the following
observation in the case of a $q-$analogue of a Steiner system: Given
a $q-$analogue of a $t-(v,k,1)$ design $C$ we get a constant dimension
code of minimum distance $2(k-t+1).$ As each $t-$dimensional space
is in only one $k-$dimensional subspace the intersection between
two spaces from $C$ is at most $(t-1)-$dimensional. Therefore the
minimum distance of $C$ is at least $2(k-t+1).$ On the other hand
given any $(t-1)-$dimensional subspace $V$ we can find two $t-$dimensional
spaces $U,W$ with intersection $V$ and then two unique $q-$blocks
containing $U$ and $W.$ The minimum distance between these $q-$blocks
is $2(k-t+1).$

$q-$analogues of designs were introduced by Thomas in 1987 \cite{thomas_q_analoga_1987}.
Later they were studied in a paper by Braun et al. \cite{Braun_kerber_laue_q_analoga}
were the authors constructed the first non-trivial $q-$analogue of
a $3-$design. We will use the methods described in their paper to
construct constant dimension space codes. 

In later papers by Thomas \cite{Thomas_qdesign_restrictionenbeilambda1_1996}
and Etzion and Schwartz \cite{etzion_qdesigns_2002} it was shown
that there are severe restrictions on the possible existence of $q-$analogues
of Steiner systems. We will search for a collection of subspaces satisfying
only the conditions given by (P1) and not for the stronger condition
satisfied by a $q-$analogue of a Steiner system. In general this
method could also be used for the search for Steiner systems.

\section{Construction of Constant Dimension Codes}

In this section we describe how to construct a constant dimension
code $C$ using a system of Diophantine linear equations and inequalities.
Due to the definition of the subspace distance for all $V,W\in C$
we have $d_{S}(V,W)=2k-2dim(V\cap W)$ where $k$ is the dimension
of the code. Thus the minimum subspace distance has to be an even
number less or equal to $2k$. To construct a constant dimension subspace
code of dimension $k$ and minimum subspace distance $2d$ we have
to find $n$ subspaces $\{ V_{1},\ldots,V_{n}\}$ of dimension $k$
such that there is no subspace of dimension $k-d+1$ contained in
two of the selected $k-$spaces. We define $M$ as the incidence matrix
of the incidence system between the $(k-d+1)-$spaces (labeling the
rows of $M$) and the $k-$spaces (labeling the columns):\[
M_{W,V}:=\left\{ \begin{array}{cc}
1 & \textrm{if}\, V\textrm{ $\textrm{contains \,}$}W,\\
0 & \textrm{otherwise.}\end{array}\right.\]
Using $M$ we get the description of a constant dimension code as
the solution of a Diophantine system. We denote by $s$ the number
of columns of $M.$

\begin{theorem}
\label{thm1:}~

There is a constant dimension code with $m$ codewords and minimum
distance at least $2d$ if and only if there is a $(0/1)-$solution
$x=(x_{1},\ldots,x_{s})^{T}$ of the following system of one equation
and one set of inequalities:

(1)~~~~~~~~~~~~~~~~~~~~ ~~~~~~~~~~~~~~~~~~~~~~~~~
$x_{1}+\ldots+x_{s}=m$,

(2)\[
\begin{array}{ccccc}
 & Mx & \le & \left(\begin{array}{c}
1\\
\vdots\\
1\end{array}\right) & .\end{array}\]

\end{theorem}
This set of equation has to be read as follows: A solution $x$ has
the property that the product of $x$ with a single row of $M$ is
$0$ or $1.$ To get the constant dimension code corresponding to
a solution we have to use the $(0/1)-$vector $x$ as the characteristic
vector of a subset of the set of all $k-$dimensional subspaces of
$GF(q)^{v}.$ Theorem \ref{thm1:} is a generalization of the Diophantine
system describing the search for a $q-$analogue of a Steiner system
which was given in \cite{Braun_kerber_laue_q_analoga}. 

\begin{corollary}
\cite{Braun_kerber_laue_q_analoga}~

There is a $q-$analogue of a $(k-d+1)-(v,k,1)$ design with $b$
blocks if and only if there is a $(0/1)-$solution $x=(x_{1},\ldots,x_{s})^{T}$
of the following system of Diophantine linear equations:

(1) ~~~~~~~~~~~~~~~~~~~~ ~~~~~~~~~~~~~~~~~~~~~~~~~$x_{1}+\ldots+x_{s}=b$,

(2)\[
\begin{array}{ccccc}
 & Mx & = & \left(\begin{array}{c}
1\\
\vdots\\
1\end{array}\right) & .\end{array}\]

\end{corollary}
The size of these problems is given by the number of subspaces in
$GF(q)^{v}$. In general this number is growing too fast. The number
of $k-$dimensional subspaces of $GF(q)^{v}$ is given by the $q-$binomial
coefficients:\[
\left[\begin{array}{c}
v\\
k\end{array}\right]_{q}:={\displaystyle \prod_{j=1..k}\frac{(1-q^{v+1-j})}{(1-q^{j})}}.\]

Already in the smallest case of a $q-$analogue of the Fano plane
for $q=2$ the matrix $M$ has $11811$ columns and $2667$ rows.

\section{Constant Dimension Codes with prescribed Automorphisms\label{sec:Constant-Dimension-Codes}}

To handle also larger cases we apply the following method. We no longer
look for an arbitrary constant dimension code. We are now only interested
in a set of spaces which have a prescribed group of automorphisms.
An automorphism $\varphi$ of set $C=\{ V_{1},\ldots,V_{m}\}$ is
an element from $GL(v,GF(q))$ such that $C=\{\varphi(V_{1}),\ldots,\varphi(V_{m})\}.$
We denote by $G$ the group of automorphisms of $C$, which is a subgroup
of $GL(v,GF(q))$.

The main advantage of prescribing automorphisms is that the size of
the system of equations is much smaller. The number of variables will
be the number of orbits of $G$ on the $k-$spaces. The number of
equations or inequalities will be the number of orbits on the $(k-d+1)-$spaces.
The construction process will then have two steps:

\begin{itemize}
\item In a first step the solution of a construction problem is described
as a solution of a Diophantine system of linear equations.
\item In a second step the size of the linear system is reduced by prescribing
automorphisms.
\end{itemize}
This construction method is a general approach that works for many
discrete structures as designs \cite{Kramer_Mesner_ihre_matrix,Betten_Kerber_Kohnert_laue_Wassermann_design},
$q$-analogs of designs \cite{Braun_q_analoga,Braun_kerber_laue_q_analoga},
arcs in projective geometries \cite{BKW_arcs_2005}, linear codes
\cite{BBFKKW06,micbra2,BraunKohnertWassermann:05a,maruta-projectivities-2008}
or quantum codes \cite{tonchev_quatum_2008}. 

The general method is as follows: The matrix $M$ is reduced by adding
up columns (labeled by the $k-$spaces) corresponding to the orbits
of $G.$ Now because of the relation\begin{equation}
W\,\mbox{subspace of}\, V\iff\varphi(W)\,\mbox{subspace of}\,\varphi(V)\label{eq:1}\end{equation}
for any $k-$space $V$ and $(k-d)-$space $W$ and any automorphism
$\varphi\in G$ the rows corresponding to lines in an orbit under
$G$ are equal. Therefore the redundant rows are removed from the
system of equations and we get a smaller matrix denoted by $M^{G}.$
The number of rows of $M^{G}$ is then the number of orbits of $G$
on the $(k-d+1)-$spaces. The number of columns of $M^{G}$ is the
number of orbits of $G$ on the $k-$spaces. We denote by $\omega_{1},\ldots$
the orbits on the $k-$spaces and by $\Omega_{1},\ldots$ the orbits
on the $(k-d+1)-$spaces. For an entry of $M^{G}$ we have:\[
M_{\Omega_{i},\omega_{j}}^{G}=|\{ V\in\omega_{j}:W\,\,\mbox{is a subspace of }V\}|\]
where $W$ is a representative of the orbit $\Omega_{i}$ of $(k-d+1)-$spaces.
Because of property \eqref{eq:1} the matrix $M$ is well-defined
as the definition of $M_{\Omega_{i},\omega_{j}}^{G}$ is independent
of the representative $W$. Now we can restate the above theorem in
a version with the condensed matrix $M^{G}:$

\begin{theorem}
~\label{thm:2}

Let $G$ be a subgroup of $GL(v,GF(q))$. There is a constant dimension
code of length $m$ and minimum distance at least $2d$ whose group
of automorphisms contains $G$ as a subgroup if, and only if, there
is a $(0/1)-$solution $x=(x_{1},\ldots)^{T}$ of the following system
of equations:

(1) ~~~~~~~~~~~~~~~~~~~~ ~~~~~~~~~~~~~~~~~~~~~~~~~$\left|\omega_{1}\right|x_{1}+\ldots=m$,

(2)\[
\begin{array}{ccccc}
 & M^{G}x & \le & \left(\begin{array}{c}
1\\
\vdots\\
1\end{array}\right) & .\end{array}\]

\end{theorem}
There is one further reduction possible. We are looking for a $(0/1)-$solution
where each inner product of a row of $M^{G}$ and the vector $x$
is less or equal to $1.$ We can remove columns of $M^{G}$ with entries
greater than $1.$ This gives a further reduction of the size of $M^{G}.$
After this last removal of columns we again check on equal rows and
also on rows containing only entries equal to zero. These all zero
rows and all but one copy of equal rows are also removed. 

In order to locate large constant dimension codes with given parameters
$q,k$ and $2d$ we try do find feasible solutions $x=(x_{1},\ldots)^{T}$
of the system of equations of Theorem \ref{thm:} for a suitable chosen
group $G$ and a suitable chosen length $m$. Here we remark that
we have the freedom to change equality (1) into \[
\sum_{i}|\omega_{i}|x_{i}\ge m.\]
 For this final step we use some software. Currently we use a variant
of an LLL based solver written by Alfred Wassermann \cite{wassermann_BMS_habil_2006}
or a program by Johannes Zwanzger \cite{zwanzger-heuristics-2007}
which uses some heuristics especially developed for applications in
coding theory. The advantage of the LLL based solver is that we definitely
know whether there exist feasible solutions or not whenever the program
runs long enough to terminate. Unfortunately for the examples of Section
\ref{sec:Results} this never happens so that practically we could
only use this solver as a heuristic to quickly find feasible solutions.

If we change equality (1) into a target function \[
f(x)=f(x_{1},\ldots)=|\omega_{i}|x_{i}\]
we obtain a formulation as a binary linear optimization problem. In
this case we can apply the commercial ILOG CPLEX 11.1.0 software for
integer linear programs. The big advantage of this approach is that
at every time of the solution process we have lower bounds, corresponding
to a feasible solution with the largest $f(x)$-value found so far,
and upper bounds on $f(x)$. 

We can even reformulate our optimization problem in the language of
graph theory. Here we consider the variable indices $i$ as vertices
of a graph $\mathcal{G}$ each having weight $\left|\omega_{i}\right|$.
The edges of $\mathcal{G}$ are implicitely given by inequality (2).
Therefore let us denote the $i$'s row of $M^{G}$ by $M_{i,\cdot}^{G}$.
Now the inequality $M_{i,\cdot}^{G}\le1$ translates into the condition
that the set\[
\mathcal{C}_{i}:=\left\{ j\,:\, M_{i,j}^{G}=1\right\} \]
 is an independence set in $\mathcal{G}$. To construct the graph
$\mathcal{G}$ we start with a complete graph and for each row $M_{i,\cdot}^{G}$
we delete all edges between vertices in $\mathcal{C}_{i}$. Now an
optimal solution of the binary linear program corresponds to a maximum
weight clique in $\mathcal{G}$. Again there exist heuristics and
exact algorithms to determine maximum weight cliques in graphs. An
available software package for this purpose is e.g. CLIQUER \cite{cliquer}.

This approach allows to use the whole bunch of clique bounds from
algebraic graph theory to obtain upper bounds on the target function
$f(x)$. In the case where we are able to locate large independent
sets in $\mathcal{G}$ which are not subsets of the $\mathcal{C}_{i}$
we can use them to add further inequalities to (2). If those independent
sets are large enough and not too many then a solver for integer linear
programs highly benefits from the corresponding additional inequalities.

For theoretical upper bounds and practical reasons how to quickly
or exhaustively locate solutions of our system of Theorem \ref{thm:2}
it is very useful to have different formulations of our problem to
be able to apply different solvers.

\subsection{Example\label{sub:Example}}

We start with the space $GF(2)^{7}.$ We now describe the construction
of a subspace code with $304$ codewords and constant dimension equal
to $3.$ This code will have minimum subspace distance $4.$ The matrix
$M$ is the incidence matrix between the $3-$dimensional subspaces
of $GF(2)^{7}$ and the $2-$dimensional subspaces. Without further
reductions this matrix has $\left[\begin{array}{c}
7\\
3\end{array}\right]_{2}=11811$ columns and $\left[\begin{array}{c}
7\\
2\end{array}\right]_{2}=2667$ rows. We prescribe now a group $G$ of automorphisms generated by
a single element:\[
G:=\left\langle \left(\begin{array}{ccccccc}
1\\
 &  &  &  & 1\\
 &  &  &  &  & 1\\
 &  &  &  &  &  & 1\\
 & 1 & 1 &  & 1 & 1\\
 &  & 1 & 1 &  & 1 & 1\\
 & 1 & 1 & 1 &  & 1 & 1\end{array}\right)\right\rangle .\]

This group $G$ has $567$ orbits on the $3-$spaces and $129$ orbits
on the $2-$spaces. Using Theorem \ref{thm:} we can formulate the
search for a large constant dimension code as a binary linear maximization
problem having $129$ constraints and $567$ binary variables. After
a presolving step, automatically performed by the ILOG CPLEX software,
there remain only $477$ binary variables and $126$ constraints with
$3306$ nonzero coefficients.

After some minutes the software founds a $(0/1)-$solution with $16$
variables equal to one. Taking the union of the corresponding $16$
orbits on the $3-$spaces of $GF(2)^{7}$ we got a constant dimension
space code with $304$ codewords having minimum distance $4.$ Previously
known was a code with $289$ codewords obtained from a construction
using rank-metric codes (\cite{silberstein-2008} p.28) and another
code consisting of $294$ subspaces discovered by A. Vardy (private
email communication).

In general it is difficult to construct the condensed matrix $M^{G}$
for an arbitrary group and larger parameters $v$ and $k$ as the
number of subspaces given by the Gaussian polynomial $\left[\begin{array}{c}
v\\
k\end{array}\right]_{q}$ grows very fast and it becomes difficult to compute all the orbits
necessary for the computation of $M^{G}.$ In the following section
we give a method to get a similar matrix in special cases.

\section{Using Singer Cycles\label{sec:Using-Singer-Cycles}}

A special case of the above method is the use of a Singer cycle. We
use for the reduction a Singer subgroup of $GL(v,q)$ which acts transitively
on the one-dimensional subspaces of $GF(q)^{v}.$ Singer cycles have
been used in many cases for the construction of interesting geometric
objects \cite{0998.20003}. We will now describe a method to construct
a set $C$ of $k-$subspaces of $GF(q)^{v}$ with the following two
special properties:

\begin{enumerate}
\item $C$ has the Singer subgroup as a subgroup of its group of automorphisms.
\item The dimension of the intersection of two spaces from $C$ is at most
one.
\end{enumerate}
Such a set $C$ of course is a constant dimension subspace code of
minimum distance $2(k-1).$ This is a special of the situation of
Theorem \ref{thm:}. We now fix one generator $g\in GL(v,q)$ of a
Singer subgroup $G$ and a one-dimensional subspace $V$ of $GF(q)^{v}.$
As $G$ acts transitively on the one-dimensional subspaces we can
label any one-dimensional subspace $W$ by the unique exponent $i$
between $0$ and $l:=\left[\begin{array}{c}
v\\
k\end{array}\right]_{2}-1$ with the property that $W=g^{i}V.$ Given a $k-$spaces $U$ we can
describe it by the set of one-dimensional (i.e. numbers between $0$
and $l$) subspaces contained in $U.$ Given such a description of
a $k-$spaces it is now easy to get all the spaces building the orbit
under the Singer subgroup $G$. You simply have to add one to each
number and you will get the complete orbit if you do this $l$ times.

\begin{example}
$q=2,v=5,k=2$ 

A two-dimensional binary subspace contains three one-dimensional subspaces.
We get a two-dimensional space by taking the two one-dimensional spaces
labeled $\{0,1\}$ and the third one given by the linear combination
of these two will have a certain number in this example $\{14\}.$
Therefore we have a two dimensional space described by the three numbers
$\{0,1,14\}.$ Two get the complete orbit under the Singer subgroup
we simply has to increase the numbers by one for each multiplication
by the generator $g$ of the Singer subgroup. The orbit length of
the Singer subgroup is $31$ and the orbit is build by the $31$ sets:
$\{0,1,14\},\{1,2,15\},\ldots,\{16,17,30\},\{0,17,18\},\ldots$ $\{12,29,30\},$
$\{0,13,30\}.$
\end{example}
To describe the different orbits of the Singer subgroup we build the
following set of pairwise distances:

Denote by $s$ the number of one-dimensional subspaces in $k-$space.
Let $\{ v_{1},\ldots,v_{s}\}\subset\{0,1,\ldots,l\}$ be the set of
$s$ numbers describing a $k-$space $U$. Denote by $d_{\{ i,j\}}$
the distance between the two numbers $v_{i}$ and $v_{j}$ modulo
the length of the Singer cycle. $d_{\{ i,j\}}$ is a number between
$1$ and $l/2.$ We define the multiset $D_{U}:=\{ d_{\{ i,j\}}:1\le i<j\le s\}$.
We call $D_{U}$ the distance distribution of the subspace $U.$ All
the spaces in an orbit of a Singer subgroup have the same distance
distribution and on the other hand different orbits have different
distance distribution. We therefore also say that $D_{U}$ is the
distance distribution of the orbit.

We use these distance distribution to label the different orbits of
the Singer subgroup of the $k-$spaces. The first observation is:

\begin{lemma}
A Singer orbit as a subspace code

An orbit $C=$ $\{ V_{0},\ldots,V_{l}\}$ of Singer subgroup on the
$k-$subspaces of $GF(q)^{v}$ is a subspace code of minimum distance
$2(k-1)$ if and only if the distance distribution of the orbit has
no repeated numbers.
\end{lemma}
\begin{proof}
We have to show that the intersection of any pair of spaces in $C$
has at most one one-dimensional space in common. Having no repeated
entry in the distance distribution means that a pair of numbers (i.e.
pair of one-dimensional subspaces) in a $q-$block $b$ of $C$ can
not be built again by shifting the numbers in $b$ using the operation
of the Singer subgroup on $b.$ 
\end{proof}
The same is true if we want to construct a subspace code by combining
several orbits of the Singer subgroup. We have to check that the intersection
between two spaces is at most one-dimensional for this we define the
matrix $S$, whose columns are labeled by the orbits $\Omega_{j}$
of the Singer subgroup on the $k-$dimensional subspaces of $GF(q)^{v}$
and the rows are labeled by the possible number $i\in\{0,\ldots,l/2\}$
in the distance distribution of the $k-$spaces. Denote by $D_{\Omega_{j}}$
the distance distribution of the $j-$th orbit then we define an entry
of the matrix $S$ by\[
S_{i,\Omega_{j}}:=\left\{ \begin{array}{cc}
1 & \mbox{if }i\in D_{\Omega_{j}}\\
0 & \textrm{otherwise.}\end{array}\right..\]
Using this matrix $S$ we have the following characterization of constant
dimension codes with prescribed automorphisms:

\begin{theorem}
~\label{thm:}

There is a constant dimension code $C$ with $n\cdot(l+1)$ codewords
and minimum distance at least $2(k-1)$ whose group of automorphisms
contains the Singer subgroup as a subgroup if, and only if, there
is a $(0/1)-$solution $x=(x_{1},\ldots)^{T}$ of the following system
of equations:

(1) $\Sigma x_{i}=n$,

(2)\[
\begin{array}{ccccc}
 & Sx & \le & \left(\begin{array}{c}
1\\
\vdots\\
1\end{array}\right) & .\end{array}\]

\end{theorem}
This is the final system of one Diophantine linear equation together
with $l/2+1$ inequalities which we successfully solved in several
cases.

\section{Results\label{sec:Results} }

As mentioned in the introduction there is an increased interest on
constant dimension codes with a large number of codewords for a given
minimum subspace distance. There are (very) recent ArXiV-preprints
\cite{etzion-silberstein-2008,etzion_vardy_isit_2008,silberstein-2008}
giving some constructions for those codes. 

Here we restrict ourselves on the binary field $q=2$ and dimension
$k=3$ and minimum subspace distance $d_{S}=4.$

Using the approach described in Section \ref{sec:Using-Singer-Cycles}
it was possible to construct constant dimension codes using the Singer
cycle with the following parameters. We denote by $n$ the number
of orbits used to build a solution, by $d$ we denote minimum space
distance of the corresponding constant dimension space code:

\begin{center}\begin{tabular}{|c|c|c|c|c|c|c|}
\hline 
$v$&
$k$&
$n$&
$\#$ orbits&
$\#$ codewords&
best&
$d$\tabularnewline
\hline
\hline 
$6$&
$3$&
$1$&
$19$&
$1\cdot63=63$&
$71$\cite{silberstein-2008}&
$4$\tabularnewline
\hline 
$7$&
$3$&
$2$&
$93$&
$2\cdot127=254$&
$294$&
$4$\tabularnewline
\hline 
$8$&
$3$&
$5$&
$381$&
$5\cdot255=1275^{*}$&
$1164$\cite{silberstein-2008}&
$4$\tabularnewline
\hline 
$9$&
$3$&
$11$&
$1542$&
$11\cdot511=5621^{*}$&
$4657$\cite{silberstein-2008}&
$4$\tabularnewline
\hline 
$10$&
$3$&
$21$&
$6205$&
$21\cdot1023=21483^{*}$&
$18631$\cite{silberstein-2008}&
$4$\tabularnewline
\hline
$11$&
$3$&
$39$&
$24893$&
$39\cdot2047=79833^{*}$&
$74531$\cite{silberstein-2008}&
$4$\tabularnewline
\hline
$12$&
$3$&
$77$&
$99718$&
$77\cdot4095=315315^{*}$&
$298139$\cite{silberstein-2008}&
$4$\tabularnewline
\hline
$13$&
$3$&
$141$&
$399165$&
$141\cdot8191=1154931$&
$1192587$\cite{silberstein-2008}&
$4$\tabularnewline
\hline
$14$&
$3$&
$255$&
$1597245$&
$255\cdot16383=4177665$&
$4770411$\cite{silberstein-2008}&
$4$\tabularnewline
\hline
\end{tabular}\par\end{center}

In \cite{etzion_vardy_isit_2008} the authors defined the number $A_{q}(v,d,k)$
as the maximal number of codewords in a constant dimension subspace
code of minimum distance $d.$ They derived lower and upper bounds.
We have implemented the construction method described in \cite{silberstein-2008}
to obtain the resulting code sizes which give the lower bounds for
$A_{q}(v,d,k)$ for $v\ge9$. In the above table we marked codes which
improved the lower bounds on $A_{q}(v,d,k)$ with an $*$. We would
like to remark that for $6\le v\le8$ our results are optimal for
the Singer cycle as a subgroup of the group of automorphisms. So far
for $v=9$ a code size of $n=12$ is theoretically possible. (In this
case the corresponding binary linear program was not solved to optimality.)

Since for $v=6,7$ the method using the Singer cycle was not capable
of beating the best known constant dimension code we tried the more
general approach described in Section \ref{sec:Constant-Dimension-Codes}.

For $v=6$ even the original incidence matrix $M$ or $M^{G}$where
$G$ is the identity group results in only 1395 binary variables and
651 constraints having 9765 nonzero entries. Using the ILOG CPLEX
11.1.0 solver directly on this problem yields a constant dimension
code of cardinality $n=74$ which beats the example of \cite{etzion-silberstein-2008,silberstein-2008}
by $3$. The upper bound in this case is given by $93$.

As mentioned in Example \ref{sub:Example} the original incidence
matrix $M$ is quite large. Here the direct approach has not lead
to any improvements. Although in general it is difficult to construct
the condensed matrix $M^{G}$ for an arbitrary group and larger parameters
we were able to conquer the difficulties for $v=7,k=3,$ $d_{S}=4$
and some groups. The group resulting in the code having $304$ three-dimensional
subspaces of $GF(2)^{7}$ such that the intersection of two codewords
has dimension at most one was already given in Example \ref{sub:Example}.
We have tried several groups before ending up with this specific group.
More details can be shown using the following diagram:

\noindent \begin{center}\includegraphics[scale=0.33]{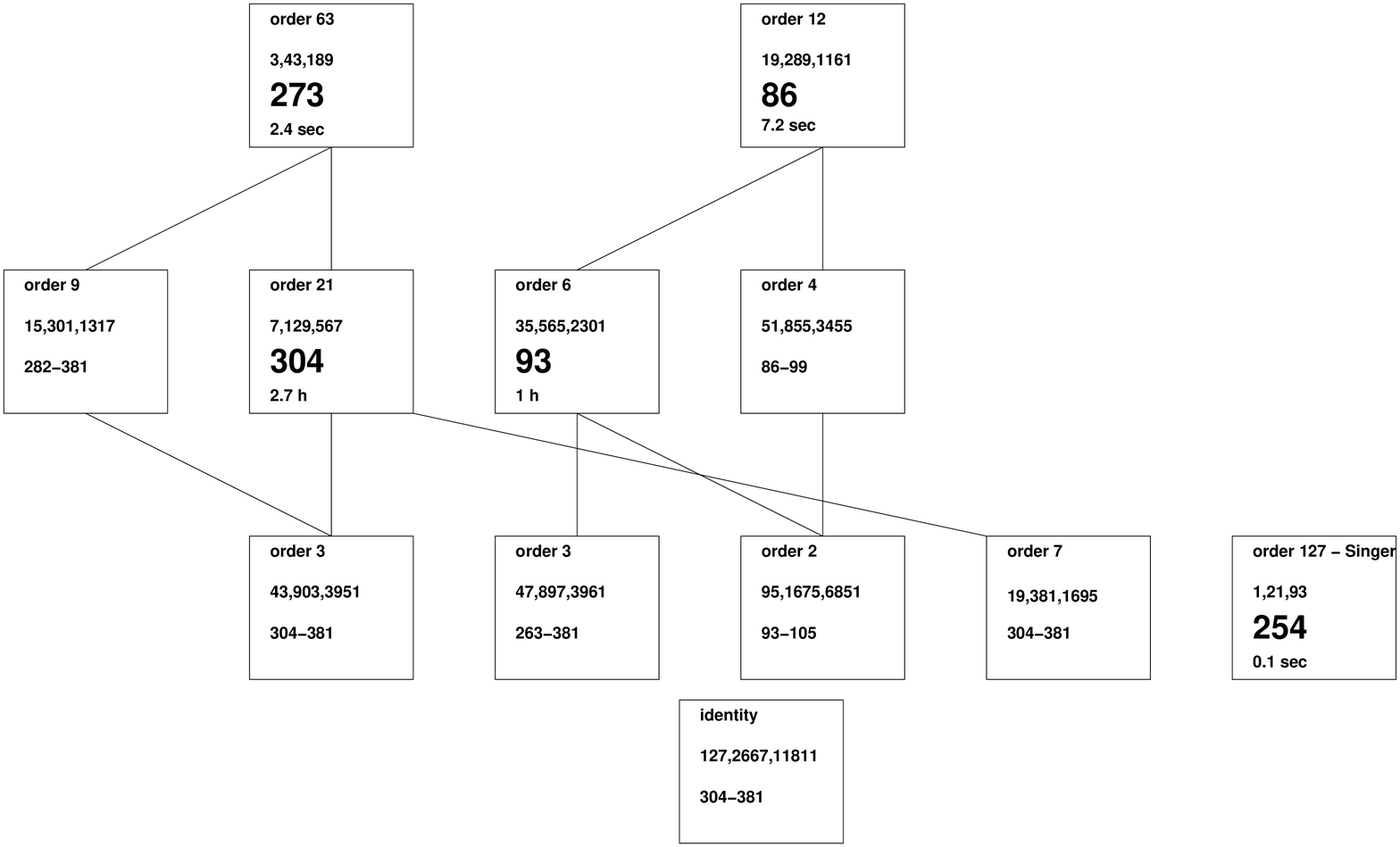}\par\end{center}

This picture shows part of the subgroup lattice of the automorphism
group $PGL(7,2)$ of the $L(GF(2)^{7}).$ It only shows cyclic groups
and in the top row we give the order of the group. In the second row
we give the number of orbits on the points, lines and planes. In the
third row of each entry we give the size $lb$ of a subspace code
and the best found upper bound $ub$ in the format $lb-ub$. As described
in Section \ref{sec:Constant-Dimension-Codes} for a given group our
problem corresponds to several versions of feasibility or optimization
problems. To obtain the lower bounds we have used the LLL based algorithm,
the coding theoretic motivated heuristic and the ILOG CPLEX solver
for integer linear programs. The upper bounds were obtained by the
CPLEX solver stopping the solution process after a reasonable time.
Whenever the lower and the upper bound meet we have written only one
number in bold face. In each of these cases we give the necessary
computation time to prove optimality in the forth row.

As we can split orbits if we move to a subgroup we can translate a
solution found for a group $G$ into a solution for a subgroup of
$G.$ E.g. for the groups of order smaller than 21 we did not find
codes of size 304 directly. This fact enables us to perform a restricted
search in systems corresponding to subgroups by only considering solutions
which are in some sense \emph{near} to such a translated solution.
We have tried this for the subgroups of the group of order $21$ -
unfortunately without success.

We would like to remark that solving the linear relaxation can prevent
other heuristics from searching for good solutions where no good solutions
can exist. E.g. we can calculate in a second that every code in the
case of the third group in the third row can contain at most $106$
codewords. Since we know better examples we can skip calculations
in this group and all groups which do contain this group as a subgroup.

Finally we draw the conclusion that following the approach described
in Section \ref{sec:Constant-Dimension-Codes} it is indeed possible
to construct good constant dimension codes for given minimum subspace
distance. Prescribing the Singer cycle as a subgroup of the automorphism
group has some computational advantages. The resulting codes are quite
competitive for $v\ge8$. The discovered constant dimension codes
for $v=6,7$ show that it pays off to put some effort in the calculation
of the condensed matrix $M^{G}$ for other groups.

\bibliographystyle{plain}
\bibliography{codes,incidence}

\end{document}